\pdfoutput=1
\documentclass[11pt]{article}
\usepackage[utf8]{inputenc}
\usepackage{fullpage}
\usepackage[ruled,vlined]{algorithm2e}
\usepackage{amsmath, amsthm, amssymb, mathtools, color, hyperref, framed}
\usepackage[noend]{algpseudocode}
\newenvironment{protocol}[1][htb]
  {% Update algorithm name
   \begin{algorithm}[#1]%
  }{\end{algorithm}}

\newcommand{\ket}[1]{|#1\rangle}
\newcommand{\bra}[1]{\langle#1|}
\newcommand{\ketbra}[2]{|#1\rangle\langle#2|}
\newcommand{\braket}[2]{\langle#1|#2\rangle}
\newcommand{\norm}[1]{\left\lVert#1\right\rVert}
\newcommand{\ceil}[1]{\lceil{#1}\rceil}

\newtheorem{theorem}{Theorem}[section]
\newtheorem{corollary}[theorem]{Corollary}

\newtheorem{lemma}[theorem]{Lemma}
\newtheorem{claim}[theorem]{Claim}

\newtheorem{definition}[theorem]{Definition}

\newcommand{\C}{\mathbb{C}}
\newcommand{\eps}{\varepsilon}
\renewcommand{\epsilon}{\varepsilon}

\newcommand{\cbra}[1]{\left\{#1\right\}}
\newcommand{\rbra}[1]{\left(#1\right)}
\newcommand{\mathify}[1]{\ifmmode{#1}\else\mbox{$#1$}\fi}

\newcommand{\zone}{\{0, 1\}}

\DeclarePairedDelimiter\abs{\lvert}{\rvert}%

\makeatletter
\let\oldabs\abs
\def\abs{\@ifstar{\oldabs}{\oldabs*}}
\makeatother

\newcommand{\Qpureoneway}[1]{\mathsf{Q}^{\mathrm{pure}, \rightarrow}_{#1}}
\newcommand{\Q}[1]{\mathsf{Q}^{\mathsf{pri}}_{#1}}
\newcommand{\rank}{\mathsf{rk}}
\newcommand{\rankeps}{\mathsf{rk}_{\eps}}
\newcommand{\ranknonneg}[1]{\mathsf{rk}^{+}_{#1}}
\newcommand{\rankpsd}[1]{\mathsf{rk}^{\mathrm{psd}}_{#1}}
\newcommand{\Rpri}[1]{\mathsf{R}^{\mathsf{pri}}_{#1}}

\newcommand{\Rpub}[1]{\mathsf{R}^{\mathsf{pub}}_{#1}}

\newcommand{\EQ}{\mathsf{EQ}}

\newcommand{\SINK}{\mathsf{SINK}}

\newcommand{\XOR}{\mathsf{XOR}}

\newcommand{\inp}[2]{\langle{#1},{#2}\rangle} % inproduct, < , >
\newcommand{\tr}{\mbox{\rm tr}}

\newcommand{\R}{\mathbb{R}}

\newcommand{\E}{\mathbb{E}}
\newcommand{\F}{\mathbb{F}}

\title{Tight Bounds for the Randomized and Quantum\\ Communication Complexities of Equality with Small Error}
\author{Olivier Lalonde\thanks{DIRO, Université de Montréal. Supported by the Natural Sciences and Engineering Research Council of Canada (NSERC). {\tt olivier.lalonde.1@umontreal.ca}} 
\and
Nikhil S.~Mande\thanks{University of Liverpool, UK. Work done while the author was a postdoc at CWI, Amsterdam, and supported by the Dutch Research Council (NWO) through QuantERA ERA-NET Cofund project QuantAlgo (project number 680-91-034). {\tt mande@liverpool.ac.uk}}
\and 
Ronald de Wolf\thanks{QuSoft, CWI and University of Amsterdam, the Netherlands. 
Partially supported by the Dutch Research Council (NWO/OCW), as part of the Quantum Software Consortium programme (project number 024.003.037), and through QuantERA ERA-NET Cofund project QuantAlgo (680-91-034). {\tt rdewolf@cwi.nl}
}
}
\date{}

\begin{document}

\maketitle

\begin{abstract}
\noindent
We investigate the randomized and quantum communication complexities of the well-studied Equality function with small error probability~$\eps$, getting the optimal constant factors in the leading terms in various different models.

The following are our results in the \emph{randomized} model:
\begin{itemize}
    \item We give a general technique to convert public-coin protocols to private-coin protocols by incurring a small multiplicative error at a small additive cost. This is an improvement over Newman's theorem [Inf.~Proc.~Let.'91] in the dependence on the error parameter.
    \item As a consequence we obtain a $(\log(n/\eps^2) + 4)$-cost private-coin communication protocol that computes the $n$-bit Equality function, to error $\eps$. This improves upon the   $\log(n/\eps^3) + O(1)$ upper bound implied by Newman's theorem, and matches the best known lower bound, which follows from Alon [Comb.~Prob.~Comput.'09], up to an additive $\log\log(1/\eps) + O(1)$.
\end{itemize}
The following are our results in various \emph{quantum} models:
\begin{itemize}
    \item We exhibit a one-way protocol with $\log(n/\eps) + 4$ qubits of communication for the $n$-bit Equality function, to error $\eps$, that uses only pure states. This bound was implicitly already shown by Nayak [PhD~thesis'99]. 
    \item We give a near-matching lower bound: any $\eps$-error one-way protocol for $n$-bit Equality that uses only pure states communicates at least $\log(n/\eps) - \log\log(1/\eps) - O(1)$ qubits.
    \item We exhibit a one-way protocol with $\log(\sqrt{n}/\eps) + 3$ qubits of communication that uses \emph{mixed} states. 
    This is tight up to additive $\log\log(1/\eps) + O(1)$, which follows from Alon'09.
    \item We exhibit a one-way entanglement-assisted protocol achieving error probability $\epsilon$ with $\ceil{\log(1/\eps)} + 1$ classical bits of communication and $\ceil{\log(\sqrt{n}/\eps)} + 4$ shared EPR-pairs between Alice and Bob. This matches the communication cost of the classical public coin protocol achieving the same error probability while improving upon the amount of prior entanglement that is needed for this protocol, which is $\ceil{\log(n/\eps)} + O(1)$ shared EPR-pairs.
\end{itemize}
Our upper bounds also yield upper bounds on the approximate rank, approximate nonnegative-rank, and approximate psd-rank of the Identity matrix. As a consequence we also obtain improved upper bounds on these measures for 
a function that was recently used to refute the randomized and quantum versions of the log-rank conjecture (Chattopadhyay, Mande and Sherif~[J.~ACM'20], Sinha and de Wolf~[FOCS'19], Anshu, Boddu and Touchette~[FOCS'19]).
\end{abstract}

\section{Introduction}

Yao~\cite{Yao79} introduced the classical model of communication complexity, and also subsequently introduced its quantum analogue~\cite{Yao93}. 
Communication complexity has important applications in several disciplines, such as lower bounds on circuits, data structures, streaming algorithms, and many other areas (see, for example, \cite{KN97,RY:cc} and the references therein). The basic model of communication complexity involves two parties, usually called Alice and Bob, who wish to jointly compute $F(x, y)$ for a known function $F : \zone^n \times \zone^n \to \zone$, where Alice holds $x \in \zone^n$ and Bob holds $y \in \zone^n$. The parties use a communication protocol agreed upon in advance to compute $F(x, y)$. They are individually computationally unbounded and the cost is the amount of communication between the parties on the worst-case input.

Consider the $n$-bit \emph{Equality} function, denoted $\EQ_n : \zone^n \times \zone^n \to \zone$ (or simply $\EQ$ when $n$ is clear from context), and defined as $\EQ_n(x, y) = 1$ iff $x = y$. This is arguably the simplest and most basic problem in communication complexity.
It is well known that its deterministic communication complexity equals $n$, which is maximal. However,
Yao~\cite{Yao79} already showed that if we allow some small constant error probability, then the communication complexity becomes much smaller.
In this paper we pin down the small-error communication complexity of \emph{Equality} in various communication models. Our bounds are essentially optimal both in terms of $n$ and in terms of the error.
While our optimal upper bounds only give small improvements over known bounds, Equality is such a fundamental communication problem that we feel it is worthwhile to pin down its complexity as precisely as possible and to find protocols that are as efficient as possible.

\subsection{Prior work}

Given a function $F : \zone^n \times \zone^n \to \zone$, define the $2^n \times 2^n$ \emph{communication matrix} of $F$, denoted $M_F$, by $M_F(x, y) = F(x, y)$.
Define the \emph{$\eps$-approximate rank} of a matrix $M$, denoted $\rankeps(M)$, to be the minimum number of rank-1 matrices needed such that their sum is $\eps$-close to $M$ entrywise (equivalently, $\rankeps(M)$ is the minimum rank among all matrices that are $\eps$-close to $M$ entrywise). If the rank-1 matrices are additionally constrained to be entrywise nonnegative, then the resulting measure is called the
\emph{$\eps$-approximate nonnegative-rank} of $M$, denoted $\ranknonneg{\eps}(M)$. By definition, $\ranknonneg{\eps}(M_F) \geq \rankeps(M_F)$.
Denote $\eps$-error randomized communication complexity by $\Rpri{\eps}(\cdot)$ when the players have access to private randomness, and $\Rpub{\eps}(F)$ when the players have access to public randomness (i.e., shared coin flips). Let $\Q{\eps}(\cdot)$ denote $\eps$-error quantum communication complexity, assuming private randomness. In all quantum communication models under consideration in this paper, except for the last one, Alice and Bob do not have access to pre-shared entanglement.

Krause~\cite{Kra96} showed the following lower bound on the randomized communication complexity of a Boolean function in terms of the approximate nonnegative-rank of its communication matrix.
\begin{theorem}[\cite{Kra96}]\label{thm: private coin approximate rank lower bound}
Let $F : \zone^n \times \zone^n \to \zone$ be a Boolean function and $\eps > 0$. Then,
\[
\Rpri{\eps}(F) \geq \log \ranknonneg{\eps}(M_F).
\]
\end{theorem}
Analogous to this, the following lower bound is known on the quantum communication complexity of a Boolean function, due to Nielsen~\cite{nielsen:thesis} and Buhrman and de Wolf~\cite{BdW01}.
\begin{theorem}[\cite{nielsen:thesis,BdW01}]\label{thm: quantum approximate rank lower bound}
Let $F : \zone^n \times \zone^n \to \zone$ be a Boolean function and let $\eps > 0$. Then,
\[
\Q{\eps}(F) \geq \frac{1}{2}\log \rankeps(M_F).
\]
\end{theorem}

A similar proof as that of~\cite{BdW01} can be used to show that the quantum communication complexity of a Boolean function is bounded below by the logarithm of its \emph{approximate psd-rank}, which we define below. Let $M$ be a matrix with nonnegative real entries.  A rank-$d$ \emph{psd-factorization} of $M$ consists of a set of $d\times d$ complex\footnote{Often this definition is restricted to real matrices. This can change the psd-rank by a constant factor, but no more than that~\cite[Section 3.3]{LWW:psdrank}.} psd matrices $A_i$ (one for each row of~$M$) and $B_j$ (one for each column of~$M$), such that for all $i,j$ we have $M_{ij}=\tr(A_i B_j)$.
The \emph{psd-rank} of $M$, denoted $\rankpsd{}(M)$, is the minimal~$d$ for which $M$ has such a psd factorization.  This notion has gained a lot of interest in areas such as semidefinite optimization, communication complexity, and others. See Fawzi et al.~\cite{FGPRT15} for an excellent survey. The \emph{$\eps$-approximate psd-rank} of $M$, which we denote by $\rankpsd{\eps}(M)$, is the minimum psd-rank among all matrices that are $\eps$-close to $M$ entrywise.

\begin{theorem}\label{thm: qeps apxpsd rank lower bound}
Let $F : \zone^n \times \zone^n \to \zone$ be a Boolean function and let $\eps > 0$. Then,
\[
\Q{\eps}(F) \geq \log \rankpsd{\eps}(M_F) + 1.
\]
\end{theorem}
For completeness, we prove this in Appendix~\ref{app: qepspsd}. It is easy to show that $\rankpsd{\eps}(M_F) \leq \ranknonneg{\eps}(M_F)$.
Alon~\cite{Alon09} showed the following bounds on the approximate rank of the Identity matrix.

\begin{theorem}[\cite{Alon09}]\label{thm: alon}
There exists a positive constant $c$ such that the following holds for all integers $n > 0$ and $1/2^{n/2} \leq \eps \leq 1/4$. Let $I$ denote the $2^n \times 2^n$ Identity matrix. Then,
\begin{align*}
    \rankeps(I) \geq \frac{cn}{\eps^2\log\rbra{\frac{1}{\eps}}}.
\end{align*}
\end{theorem}

Note that the $2^n \times 2^n$ Identity matrix is the communication matrix of the $n$-bit Equality function. Theorems~\ref{thm: private coin approximate rank lower bound} and~\ref{thm: alon} thus imply that for $1/2^{n/2} \leq \eps \leq 1/4$,
\begin{equation}\label{eqn: rpri eq lower bound}
\Rpri{\eps}(\EQ_n) \geq \log\rbra{\frac{n}{\eps^2}} - \log \log \rbra{\frac{1}{\eps}} - O(1).
\end{equation}

Newman~\cite{New91} proved the following theorem that shows that public-coin protocols can be converted to private-coin protocols with an additive error, with a small additive cost in the communication. For the following form, see for example,~\cite[Claim 3.14]{KN97}.

\begin{theorem}[{cf.~\cite[Claim 3.14]{KN97}}]\label{thm: newman additive}
Let $F : \zone^n \times \zone^n$ be a Boolean function. For every $\delta > 0$ and every $\eps > 0$,
\[
\Rpri{\eps + \delta}(F) \leq \Rpub{\eps}(F) + \log\rbra{\frac{n}{\delta^2}} + O(1).
\]
\end{theorem}

\noindent
Relabeling variables, Theorem~\ref{thm: newman additive} is equivalent to
\[
\Rpri{\eps(1+\delta)}(F) \leq \Rpub{\eps}(F) + \log\rbra{\frac{n}{\eps^2\delta^2}} + O(1).
\]

\subsection{Our results}
In this section we list our results, first those for randomized communication complexity, and then those for quantum communication complexity. 

\subsubsection{Randomized communication complexity}

We give an improved version of Newman's theorem (Theorem~\ref{thm: newman additive}), which allows us to convert a public-coin protocol to a private-coin one with an optimal dependence on the error. Our proof follows along similar lines as that of Newman's. Our key deviation is that we use a multiplicative form of the Chernoff bound, where previously an additive version was used.

\begin{theorem}\label{thm: improved newman new}
Let $F : \zone^n \times \zone^n \to \zone$ be a Boolean function. For all $\eps \in [0, 1/2)$ and all $\delta \in (0, 1]$,
\[
\Rpri{\eps(1+\delta)}(F) \leq \Rpub{\eps}(F) + \log \rbra{\frac{n}{\eps}}+ \log\rbra{\frac{6}{\delta^2}}.
\]
\end{theorem}

To compare Theorem~\ref{thm: newman additive} and Theorem~\ref{thm: improved newman new}, consider the $(1/n)$-error private-coin randomized communication complexity of $\EQ_n$. The $\eps$-error \emph{public-coin} communication complexity of $\EQ_n$ is at most $\log(1/\eps)$ (and this can be shown to be tight up to an additive constant).
Thus, Theorem~\ref{thm: newman additive} can at best give an upper bound of
\[
\Rpri{1/n}(\EQ_n) \leq \log n + \log (n^3) + O(1) = 4\log n + O(1).
\]
Equation~\eqref{eqn: rpri eq lower bound} implies a non-matching lower bound $\Rpri{1/n}(\EQ_n) \geq 3\log n - \log \log n - O(1)$. On the other hand, Theorem~\ref{thm: improved newman new} implies a tight upper bound (up to the additive $\log\log n + O(1)$ term) of $3\log n + O(1)$ on the $(1/n)$-error private-coin communication complexity of $\EQ_n$, by converting the $\log(1/\eps)$-cost public-coin protocol for $\EQ_n$ to a private-coin protocol.

\begin{theorem}\label{thm: rand upper bound}
For all positive integers $n > 0$ and for all $\eps \in [0, 1/2)$,
\begin{align*}
    \Rpri{\eps}(\EQ) & \leq \log\rbra{\frac{n}{\eps^2}} + 4.
\end{align*}
\end{theorem}

This shows that Alon's theorem (Theorem~\ref{thm: alon}) is tight up to the $O(\log(1/\eps))$ factor, not only for approximate rank, but also for communication complexity.
Theorem~\ref{thm: rand upper bound} and Theorem~\ref{thm: private coin approximate rank lower bound} also imply that the approximate-rank lower bound in Theorem~\ref{thm: alon} is nearly tight even restricting to \emph{nonnegative} approximations to the Identity matrix.

\begin{corollary}\label{cor: our apx nonneg rank identity upper bound}
Let $n > 0$ be an integer, and let $I$ denote the $2^n \times 2^n$ Identity matrix. Then for all $\eps \in [0, 1/2)$,
\[
\ranknonneg{\eps}(I) \leq \frac{16n}{\eps^2}.
\]
\end{corollary}

To compare the performance of Theorem~\ref{thm: newman additive} with that of Theorem~\ref{thm: improved newman new} in a more general setting, we consider the natural problem of converting a public-coin protocol to a private-coin protocol while allowing the error to double.
Setting $\delta = \eps$ in Theorem~\ref{thm: newman additive} and relabeling parameters, we obtain
\[
\Rpri{\eps}(F) \leq \Rpub{\eps/2}(F) + \log\rbra{\frac{n}{\eps^2}} + O(1).
\]
However, Theorem~\ref{thm: improved newman new} yields the following improved dependence on $\eps$ by setting $\delta = 1$ and relabeling parameters.
\[
\Rpri{\eps}(F) \leq \Rpub{\eps/2}(F) + \log\rbra{\frac{n}{\eps}} + 4.
\]

\subsubsection{Quantum communication complexity}

Prior to this work, the best known lower bound on the $\eps$-error quantum communication complexity of Equality was $\Omega(\log(n/\eps))$~\cite[Proposition~3]{BdW01}, with a constant hidden in the $\Omega(\cdot)$ that is less than~$1/2$. 
The known Theorem~\ref{thm: quantum approximate rank lower bound} and Theorem~\ref{thm: alon} immediately imply that 
\begin{equation}\label{eqn: quantum lower bound apxrank}
\Q{\eps}(\EQ_n) \geq \log\rbra{\frac{\sqrt{n}}{\eps}} - \log \log \rbra{\frac{1}{\eps}} - O(1).
\end{equation}
In terms of upper bounds, we exhibit a \emph{one-way} quantum communication upper bound with an optimal dependence on $\eps$, that uses only pure-state messages (and hence does not use even private randomness). In particular, by choosing $\eps$ to be an arbitrary small polynomial in the input size, this implies that the factor of $1/2$ in Theorem~\ref{thm: quantum approximate rank lower bound} cannot be improved when $F = \EQ_n$. Let $\Qpureoneway{\eps}(F)$ be the $\eps$-error quantum communication complexity of $F$, where the protocols are one-way and Alice is only allowed to send a pure state to Bob. We show the following.

\begin{theorem}\label{thm: quant upper bound pure}
For all positive integers $n > 0$ and for all $\eps \in [0, 1/2)$,
\[
\Qpureoneway{\eps}(\EQ_n) \leq \log \rbra{\frac{n}{\eps}} + 4.
\]
\end{theorem}

The proof uses the probabilistic method to analyze random linear codes. Nayak~\cite{Nay99} already used the same upper bound technique to show an upper bound on the bounded-error one-way quantum communication complexity of $\EQ_n$. They did not explicitly derive this error-dependence, but it follows immediately from their construction by plugging in codes with length $O(n/\eps)$ and relative distance $1/2-\sqrt{\eps}$ in~\cite[pp.16--17]{Nay99}.
We also show that this is nearly tight:

\begin{theorem}\label{thm: quant lower bound pure 1way}
There exists an absolute constant $c$ such that the following holds.
For all positive integers $n > 0$ and for all $\eps \in [1/2^n, 1/4]$,
\[
\Qpureoneway{\eps}(\EQ_n) \geq \log \rbra{\frac{n}{\eps}} - \log\log \rbra{\frac{1}{\eps}} - c.
\]
\end{theorem}

While the pure-state protocol of Theorem~\ref{thm: quant upper bound pure} has optimal dependence on~$\eps$ (up to the additive  $\log\log(1/\eps)$ term), it does not match the $n$-dependence of the lower bound of Equation~\eqref{eqn: quantum lower bound apxrank}; in fact, one-way pure-state protocols cannot match this (Theorem~\ref{thm: quant lower bound pure 1way}). However, if we allow one-way \emph{mixed-state} messages, then we can give a better upper bound and close the gap:

\begin{theorem}\label{thm: quant upper bound mixed}
For all positive integers $n > 0$ and for all $\eps \in [0, 1/2)$,
\[
\Q{\eps}(\EQ_n) \leq \log \rbra{\frac{\sqrt{n}}{\eps}} + 3.
\]
\end{theorem}

An upper bound of $\log \sqrt{n} + O(1)$ was already proved by Winter~\cite{winter:ident} for the case of constant~$\eps$, and here we obtain the correct dependence also for subconstant~$\eps$.
Our proof is again probabilistic, using known concentration properties of overlaps of random projectors to allow us to show the existence of appropriate mixed-state messages for Alice and appropriate measurements for Bob.
Theorems~\ref{thm: qeps apxpsd rank lower bound} and~\ref{thm: quant upper bound mixed} also imply upper bounds on the $\eps$-approximate psd-rank of the Identity matrix.

\begin{corollary}\label{cor: our apx psd rank identity upper bound}
Let $n > 0$ be an integer, and let $I$ denote the $2^n \times 2^n$ Identity matrix. Then for all $\eps \in [0, 1/2)$,
\[
\rankpsd{\eps}(I) \leq \frac{4\sqrt{n}}{\eps}.
\]
\end{corollary}

As noted by Lee, Wei and de Wolf~\cite[Theorem~17]{LWW:psdrank}, Alon's approximate rank lower bound (Theorem~\ref{thm: alon}) almost immediately gives a lower bound of   
$\rankpsd{\eps}(I)=\Omega\left(\frac{\sqrt{n}}{\eps\sqrt{\log(1/\eps)}}\right)$. This shows that our upper bound in Corollary~\ref{cor: our apx psd rank identity upper bound} is tight up to a multiplicative $O(\sqrt{\log(1/\eps)})$ factor.

We may also consider the amount of entanglement needed to compute $\EQ_n$ in the entanglement-assisted setting, where Alice and Bob send classical bits but share an arbitrary input-independent state $\ket{\psi}$ at the start of the protocol, for instance many EPR-pairs. Since entanglement may be used to generate shared randomness by measuring, the classical public-coin protocol yields an entanglement-assisted protocol using $\ceil{\log 1/\eps} + 1$ bits of communication and $\ceil{\log n/\epsilon} + O(1)$ shared EPR-pairs. We improve on the amount of shared entanglement that's needed by showing:

\begin{theorem}\label{thm: quant upper bound entangled}
For all positive integers $n > 0$ and for all $\eps \in (0, 1/2)$, there exists a one-way protocol for $\EQ_n$ with error probability at most $\epsilon$ using $\ceil{\log 1/\eps} + 1$ bits of communication and $\ceil{\log \sqrt{n}/\eps} + 4$ shared EPR-pairs.
\end{theorem}
We do not know if the amount of classical communication used by our protocol to achieve error probability $\epsilon$ is essentially optimal. Nor do we know if the number of EPR-pairs can be reduced further given this amount of communication.

\section{Preliminaries}

All logarithms in this paper are taken to base 2. We use $\exp(x)$ to denote $e^{x}$, where $e$ denotes Euler's number. 
For strings $x, y \in \zone^n$, define their Hamming distance by $d(x, y) := |\cbra{i \in [n] : x_i \neq y_i}|$.
For an event $X$, let $I(X)\in\zone$ denote the \emph{indicator} of $X$, which is~1 iff $X$ occurs.

\begin{definition}[Linear code]
For integers $N \geq n$, a linear code is a linear function $C : \zone^{n} \to \zone^N$.
\end{definition}
One may view a linear code as an $N \times n$ matrix $M$ over $\mathbb{F}_2$; an input $x \in \zone^n$ is mapped to $N$-bit codeword $Mx$ (where the matrix product is taken over $\F_2$).
Choosing a random linear code corresponds to choosing an $M$ with uniformly random binary entries.

We use the following well-known multiplicative form of the Chernoff bound~\cite[Theorem~4.4]{MU05}.
\begin{lemma}\label{lem: chernoff}
Let $Z_1, \dots, Z_n$ be independent random variables taking values in $\zone$. Let $Z = \sum_{i = 1}^n Z_i$, and let $\mu = \E[Z]$. Then for all $\delta \in [0, 1]$,
\begin{align*}
\Pr[Z \geq (1+\delta)\mu] & \leq \exp(-\delta^2\mu/3),\\
\Pr[Z \leq (1-\delta)\mu] & \leq \exp(-\delta^2\mu/2).
\end{align*}
\end{lemma}
We refer the reader to~\cite{KN97,RY:cc} for the basics of classical communication complexity, to \cite{nielsen&chuang:qc} for quantum computing,  and to~\cite{Wol02} for an introduction to quantum communication complexity.
\section{An improved form of Newman's theorem}

\begin{proof}[Proof of Theorem~\ref{thm: improved newman new}]
Let $\Pi$ be a public-coin protocol that computes $F$ with error $\eps$. Assume without loss of generality that all the random coins are tossed at the beginning of the protocol. That is, for every $x, y \in \zone^n$,
\begin{equation}\label{eqn: pi error probability}
\Pr_r[\Pi(x, y, r) \neq F(x, y)] \leq \eps.
\end{equation}
Set 
\begin{equation}\label{eqn: choice of B}
B = \frac{6n}{\delta^2 \eps}
\end{equation}
and independently choose random strings $r_1, \dots, r_B$ according to the same distribution as used by~$\Pi$.
For two strings $x, y \in \zone^n$ and an index $j \in [B]$, let $I_{j, x, y}$ denote the indicator event of $r_j$ being a ``bad random string'' for $x, y$: 
\begin{equation}\label{eqn: definition ijxy new}
I_{j, x, y} := 
\begin{cases}
1 & \Pi(x, y, r_j) \neq f(x, y)\\
0 & \text{otherwise.}
\end{cases}
\end{equation}
Fix two arbitrary strings $x, y \in \zone^n$. Equation~\eqref{eqn: pi error probability} implies $\Pr_{r_1, \dots, r_B, j \in [B]}[I_{j, x, y} = 1] \leq \eps$. By linearity of expectation and our choice of $B$ in Equation~\eqref{eqn: choice of B},
\[
\E_{r_1, \dots, r_B}\left[\sum_{j \in [B]} I_{j, x, y}\right] \leq B\eps = \frac{6n}{\delta^2}.
\]
We now give an upper bound on $\Pr_{r_1, \dots, r_B}\left[\sum_{j \in [B]} I_{j, x, y} \geq B\eps(1+\delta) \right]$. Assume without loss of generality that $\Pr_{r_1, \dots, r_B, j \in [B]}[I_{j, x, y} = 1] = \eps$, and hence $\E_{r_1, \dots, r_B}\left[\sum_{j \in [B]} I_{j, x, y}\right] = B\eps$ (since the desired probability could only be smaller otherwise).
By a Chernoff bound (Lemma~\ref{lem: chernoff}),
\[
\Pr_{r_1, \dots, r_B}\left[\sum_{j \in [B]} I_{j, x, y} \geq B\eps(1+\delta) \right] \leq \exp\rbra{-\frac{\delta^2 \cdot 6n}{3 \delta^2}} = \exp(-2n) < 2^{-2n}.
\]
By a union bound over all $x, y \in \zone^n$,
\begin{align*}
\Pr_{r_1, \dots, r_B}\left[\sum_{j \in [B]} I_{j, x, y} \geq B\eps(1+\delta) ~\text{for some}~x, y \in \zone^n\right] & \leq \sum_{x, y \in \zone^n} \Pr_{r_1, \dots, r_B}\left[\sum_{j \in [B]} I_{j, x, y} \geq B\eps(1+\delta) \right]\\
& <2^{2n}\cdot 2^{-2n} = 1.
\end{align*}
Hence there exists a choice of $r_1, \dots, r_B$ such that the following holds for all $x, y \in \zone^n$:
\begin{equation}\label{eqn: union bound all errors small new}
\sum_{j \in [B]}I_{j, x, y} < B\eps(1+\delta).
\end{equation}
Fixing this choice of $r_1, \dots, r_B$, Protocol~\ref{algo: private coin protocol new} gives a private-coin protocol for $F$.

\smallskip

\begin{protocol}\label{algo: private coin protocol new}
\SetAlgoLined
    \begin{enumerate}
        \item Alice samples $j \in [B]$ uniformly at random, and sends it to Bob.
        \item Alice and Bob perform the public-coin protocol $\Pi$ assuming $r_j$ was the public random string.
    \end{enumerate}
\caption{A private-coin protocol for $F$}
\end{protocol}
To show the correctness of this protocol, our choice of $B$ (Equation~\eqref{eqn: choice of B}) and Equations~\eqref{eqn: definition ijxy new} and~\eqref{eqn: union bound all errors small new} imply that for all $x, y \in \zone^n$,
\begin{align*}
    \Pr_{j \in [B]}[\Pi(x, y, r_j) \neq f(x, y)] & < \frac{B\eps(1+\delta)}{B} = \eps(1+\delta).
\end{align*}
Hence the protocol has error probability less than $\eps(1+\delta)$. The cost of the first step of the protocol is $\log B$, and the cost of the second step is at most $\Rpub{\eps}(F)$. Thus, we have,
\[
\Rpri{\eps(1+\delta)}(F) \leq \Rpub{\eps}(F) + \log B = \Rpub{\eps}(F) + \log\frac{6n}{\delta^2\eps} = \Rpub{\eps}(F) + \log\frac{n}{\eps} + \log\frac{6}{\delta^2}.
\]
Note that if $\Pi$ was a one-way protocol, then Protocol~\ref{algo: private coin protocol new} is a one-way private-coin protocol.
\end{proof}

\section{Communication complexity upper bounds}

In this section we show randomized and quantum communication upper bounds for Equality.

\subsection{Randomized upper bound}\label{sec: randomized communication}

As an application of Theorem~\ref{thm: improved newman new}, we recover an optimal small-error private-coin communication complexity upper bound for $\EQ_n$ from a naive public-coin protocol of cost $\log(2/\eps)$ and error $\eps/2$:

\begin{equation}\label{eqn: equality optimal randomized upper bound}
\Rpri{\eps}(\EQ_n) \leq \log\rbra{\frac{2}{\eps}} + \log\rbra{\frac{n}{\eps}} + 3 = \log\rbra{\frac{n}{\eps^2}} + 4.
\end{equation}
This proves Theorem~\ref{thm: rand upper bound}.
In contrast, Newman's theorem (Theorem~\ref{thm: newman additive}) would only give an upper bound of
\[
\Rpri{\eps}(\EQ_n) \leq \log\rbra{\frac{2}{\eps}} + \log\rbra{\frac{n}{\eps^2}} + O(1) = \log\rbra{\frac{n}{\eps^3}} + O(1).
\]
In particular, for $\eps=1/n$ we improve the upper bound from $4\log n+O(1)$ to $3\log n+O(1)$, which turns out to be essentially optimal.

\subsection{Quantum upper bound with only pure states}

We require the following property of random linear codes.
\begin{claim}\label{claim: quant random code good}
Let $n$ be a positive integer and let $\delta > 0$. Let $x \neq y \in \zone^n$ be two arbitrary but fixed strings. Let $N = 4n/\delta^2$. Let $C : \zone^n \to \zone^N$ be a random linear code. Then
\[
\Pr_C\left[\frac{d(C(x),  C(y))}{N} \notin \left[\frac{1}{2} - \delta, \frac{1}{2} + \delta\right]\right] < 2^{-2n}.
\]
\end{claim}

\begin{proof}[Proof of Claim~\ref{claim: quant random code good}]
For each $i \in [N]$, the random variable $Z_i := I[C(x)_i = C(y)_i]$ equals 1 with probability $1/2$ and $0$ with probability $1/2$. Further, $Z_i$ and $Z_j$ are independent for all $i \neq j \in [N]$. Define $Z = \sum_{i = 1}^N Z_i = d(C(x), C(y))$. We have $\E[Z] = N/2$. By a Chernoff bound (Lemma~\ref{lem: chernoff}),
\[
\Pr_C\left[\abs{\frac{d(C(x),  C(y))}{N} - \frac{1}{2}} \geq \delta\right] = \Pr_C\left[\abs{Z - \frac{N}{2}} \geq 2\delta \cdot \frac{N}{2}\right] \leq 2\exp\rbra{-4\delta^2N/6} < 2^{-2n},
\]
where the last inequality holds by our choice of $N$.
\end{proof}

By a union bound over all $x, y \in \zone^n$, Claim~\ref{claim: quant random code good} implies the following corollary.

\begin{corollary}\label{cor: quant good code}
Let $n$ be a positive integer, let $\delta > 0$ and let $N = 4n/\delta^2$. Then there exists a linear code $C : \zone^n \to \zone^N$ such that for all $x \neq y \in \zone^n$,
\[
\frac{d(C(x), C(y))}{N} \in \left[\frac{1}{2} - \delta, \frac{1}{2} + \delta\right].
\]
\end{corollary}

We now prove Theorem~\ref{thm: quant upper bound pure}.
\begin{proof}[Proof of Theorem~\ref{thm: quant upper bound pure}]
Set $\delta = \sqrt{\eps}/2$. Let $N = 4n/\delta^2 = 16n/\eps$ and let $C : \zone^n \to \zone^{16n/\eps}$ be the code obtained from Corollary~\ref{cor: quant good code}. The following is a protocol for $\EQ_n$.
\begin{enumerate}
    \item Alice, on input $x\in \zone^n$ prepares state $\ket{\phi_x} :=  \frac{1}{\sqrt{N}} \sum_{i \in [N]} (-1)^{C(x)_i} \ket{i}$, and sends Bob $\ket{\phi_x}$.
    \item Define $\ket{\phi_y} := \frac{1}{\sqrt{N}} \sum_{i \in [N]} (-1)^{C(y)_i} \ket{i}$. Bob measures with respect to the projectors $\ket{\phi_y}\bra{\phi_y}$ and $I - \ket{\phi_y}\bra{\phi_y}$, and outputs 1 on observing the first measurement outcome, and 0 otherwise.
\end{enumerate}
This protocol succeeds with probability 1 when $x = y$. The only error arises when $x \neq y$ and Bob observes the first measurement outcome.
Thus, the error probability of this protocol equals 
\begin{align*}
    \max_{x \neq y \in \zone^n}|\braket{\phi_x}{\phi_y}|^2 &= \max_{x \neq y \in \zone^n} \left(\frac{1}{N} \sum_{i \in [N]} (-1)^{C(x)_i + C(y)_i}\right)^2\\
    & = \max_{x \neq y \in \zone^n} \left(1 -\frac{2d(C(x), C(y))}{N}\right)^2\\
    & \leq 4\delta^2 = \eps,
\end{align*}
where the last inequality follows from Corollary~\ref{cor: quant good code} and the last equality follows from our choice of~$\delta$.
The number of qubits sent from Alice to Bob is $\log N = \log(16n/\eps) = \log(n/\eps) + 4$.
\end{proof}
We show in Section~\ref{sec: 1way quantupure lower bound} that the protocol in the previous proof is nearly optimal if one restricts to one-way communication with only pure states.

\subsection{Quantum upper bound with mixed states}\label{sec:upperboundmixed}

In the last section we gave a $\log(n/\eps) + O(1)$ quantum upper bound on the $\eps$-error communication complexity of $\EQ_n$, where Alice was only allowed to send a pure state to Bob.  In this section we show that allowing Alice to send a \emph{mixed} state to Bob gives a communication upper bound that is better by a factor of~2 (which is in fact optimal). An upper bound of $\log \sqrt{n} + O(1)$ was already proved by Winter~\cite{winter:ident} for the case of constant~$\eps$, but here we obtain the correct dependence also for subconstant~$\eps$. Our protocol is based on concentration properties of overlaps of random projectors.
 
Consider two rank-$r$ projectors $P$ and $Q$ acting on $\C^d$.  The largest possible inner product~$\tr(PQ)$ between them is~$r$, which occurs iff $P=Q$.  However, when one or both of the projectors are Haar-random, then we expect their inner product to be much smaller, namely only $r^2/d$. This is because if we take the spectral decompositions $P=\sum_{i=1}^r\ketbra{u_i}{u_i}$ and $Q=\sum_{j=1}^r\ketbra{v_j}{v_j}$, then
$$
\tr(PQ)=\sum_{i,j=1}^r|\inp{u_i}{v_j}|^2,
$$
and the expected squared inner product between a random $d$-dimensional unit vector~$u_i$ and any fixed unit vector~$v_j$, is $1/d$.
Hayden, Leung and Winter~\cite[Lemma III.5]{HLW06} showed that this inner product is very tightly concentrated around its expectation.
\begin{claim}[{\cite[Lemma III.5]{HLW06}}]\label{claim: random projector good}
Let $P$ and $Q$ be rank-$r$ projectors on $\C^d$, where $P$ is random\footnote{More precisely, $P$ is a projection onto a uniformly chosen $r$-dimensional subspace from all $r$-dimensional subspaces of $\C^d$. We do not elaborate more on this here since it is not relevant for us.} and $Q$ is fixed. Let $\delta \in [0, 1]$.
Then
\[
\Pr\left[\tr(PQ)\geq \frac{(1+\delta)r^2}{d}\right]\leq \exp\left(\frac{-r^2\delta^2}{5}\right) < 2^{-r^2\delta^2/5}.
\]
\end{claim}
The following corollary then follows by setting parameters suitably.

\begin{corollary}\label{cor: exists good projectors}
For every integer $n > 0$ and all $\eps \in [0, 1/2)$, there exists a set $\cbra{P_x : x \in \zone^n}$ of $2^n$ rank-$r$ projectors on $\C^d$, with $r = \sqrt{10n}$ and $d = 2r/\eps$, such that $\tr(P_x P_y) < \eps r$ for all $x \neq y \in \zone^n$.
\end{corollary}

\begin{proof}
Fix $\delta = 1$ and choose rank-$r$ projectors $\cbra{P_x : x \in \zone^n}$  independently and uniformly at random. Claim~\ref{claim: random projector good} and our choice of parameters implies that for all $x \neq y \in \zone^n$,
\[
\Pr\left[\tr(P_xP_y)\geq \frac{2r^2}{d}\right] = \Pr\left[\tr(P_xP_y)\geq \eps r\right] < 2^{-r^2\delta^2/5} = 2^{-2n}.
\]
The corollary now follows by applying a union bound over all distinct $x, y \in \zone^n$.
\end{proof}
We now prove Theorem~\ref{thm: quant upper bound mixed}.

\begin{proof}
Let $\cbra{P_x : x \in \zone^n}$ be projectors on $\C^d$ as guaranteed by Corollary~\ref{cor: exists good projectors}, each of rank $r=\sqrt{10n}$, with $d = 2\sqrt{10n}/\eps$.
Our protocol for $\EQ_n$ is Protocol~\ref{algo: quant coin protocol mixed} below.

\begin{protocol}\label{algo: quant coin protocol mixed}
\SetAlgoLined
    \begin{enumerate}
        \item Alice, on input $x \in \zone^n$, sends the $\log d$-qubit mixed state $\rho_x := P_x/r$ to Bob.
        \item Bob, on input $y \in \zone^n$, measures w.r.t.~projectors $P_y, I-P_y$, and outputs 1 on observing the first measurement outcome, and 0 otherwise.
    \end{enumerate}
\caption{A mixed-state protocol $\Pi$ for $F$}
\end{protocol}
To see the correctness of this protocol, first observe that if $x = y$, then the protocol outputs the correct answer with probability~1 because $\tr(P_x\rho_x)=\tr(P_x)/r=1$. If $x \neq y$, then the error probability is the probability of Bob observing the first measurement outcome, which is
\begin{align*}
    \Pr[\Pi(x, y) \neq \EQ_n(x, y)] = \tr(P_y \rho_x) = \tr(P_yP_x)/r < \eps,
\end{align*}
from Corollary~\ref{cor: exists good projectors}.  The cost is $\log d = \log(2\sqrt{10n}/\eps) \leq \log(\sqrt{n}/\eps) + 3$ qubits of communication.
\end{proof}

\subsection{Entanglement-assisted quantum upper bounds}

We use the probabilistic method to argue the existence of a good entanglement-assisted protocol. In the following, $m \leq d$ are natural numbers to be determined later. We take the initial entangled state to be the maximally entangled state in $D=2^d$ dimensions, i.e., $d$ EPR-pairs:
\[
\ket{\Psi_{AB}} = \frac{1}{\sqrt{D}} \sum_{i \in \{0,1\}^d} \ket{i}_A \ket{i}_B.
\]
For every $z \in \{0,1\}^n$, pick independently a Haar-random element $U_z = \{\ket{\psi_{z,r}}\}_{r \in \{0,1\}^d}$ of $SU(D)$ (i.e., a random orthonormal basis is used for the $2^d$ columns of $U_z$). The following is our protocol for $\EQ_n$:

\smallskip

\begin{protocol}\label{algo: quant coin protocol entanglement assisted}
\SetAlgoLined
    \begin{enumerate}
        \item Alice, on input $x \in \zone^n$, measures her part of $\ket{\Psi}$ in the basis $U_x$, obtaining $r^A \in \{0,1\}^d$. She then sends $b \equiv r^A_1 r^A_2\ldots r^A_m$ to Bob (i.e., the first $m$ bits of $r^A$).
        \item Bob, on input $y \in \zone^n$, measures his part of $\ket{\Psi}$ in the conjugate basis of $U_y$, obtaining $r^B \in \{0,1\}^d$. He outputs 1 if $r^B_i = b_i$ for every $1 \leq i \leq m$, and he outputs~0 otherwise.
    \end{enumerate}
\caption{An entanglement-assisted protocol $\Pi'$ for $F$}
\end{protocol}

The one-way communication complexity of this protocol $\Pi'$ is $m$ bits. We proceed with its error analysis. After step~1, by properties of the maximally entangled state, the new joint state will be
\[\ket{\Psi'} = \ket{\psi_{x,r^A}}_A \otimes \overline{\ket{\psi_{x,r^A}}}_B\]
In particular, if $x=y$, then $r^A = r^B$ and the protocol is guaranteed to succeed. Suppose now that $x \neq y$. For $b \in \{0,1\}^m$, using the shorthand
\[
R_b \equiv \{r \in \{0,1\}^d \mid r_i = b_i \; \forall i \in [m]\},
\] 
we find that the probability that the protocol fails (i.e., outputs~1) is given by 
\[
\frac{1}{D} \sum_{b \in \{0,1\}^m}  \sum_{r^A,r^B \in R_b} |\braket{\psi_{x,r^A}}{\psi_{y,r^B}}|^2.
\]
Since $R_b$ has cardinality $2^{d-m}$ and the expectation over the choice of $U_z$'s of every term in the sum is $2^{-d}$, we find that the expectation of the entire sum is $2^{-m}$. The rest of our analysis will rely on the following concentration inequality, which is derived in \cite[Chapter 3]{meckes:haarnotes}:

\begin{theorem}
Let $F: SU(n) \to \mathbb{R}$ be a function with Lipschitz constant $K$ with respect to the Frobenius norm,\footnote{This means $|F(U)-F(U')|\leq K\cdot d(U,U')$ for all $U,U' \in SU(n)$, where $SU(n)$ is the group of $n\times n$ unitary matrices with determinant~1, the Frobenius norm $\norm{A}_F$ of a matrix $A$ is defined as $\sqrt{\sum_{i,j}|A_{ij}|^2}$, and the Frobenius distance is defined as $d(U,U')=\norm{U-U'}_F$.} and let $\mu$ be the uniform distribution (Haar measure) on $SU(n)$. Then, for every $\delta > 0$,
\begin{align*}
\Pr_\mu[|F(U)-\mathbb{E}_\mu[F]| > \delta] < 2\exp\left(-\frac{\delta^2 n}{4K^2}\right).
\end{align*}
\end{theorem}

We show:
\begin{theorem}
Let $\{\phi_r\}_{r \in \{0,1\}^d}$ be a fixed orthonormal basis of $\mathbb{C}^D$. Given $U = \{\psi_r\}_{r \in \{0,1\}^d} \in SU(D)$, define $F: SU(D) \to \mathbb{R}$ 
by
\[
F(U) = \sum_{b \in \{0,1\}^m}  \sum_{r,r' \in R_b} |\braket{\phi_r}{\psi_{r'}}|^2.
\]
Then $F(U)$ has Lipschitz constant $\sqrt{D}$. 
\end{theorem}

\begin{proof}
Let $U = \{\psi_r\}_{r \in \{0,1\}^d}$ and $U' = \{\psi'_r\}_{r \in \{0,1\}^d}$ be two different elements of $SU(D)$. For $b \in \{0,1\}^m$, write
\[
P_b = \sum_{r \in R_b} \ketbra{\phi_r}{\phi_r},~~~~Q_b = \sum_{r \in R_b} \ketbra{\psi_r}{\psi_r},~~~~Q'_b = \sum_{r \in R_b} \ketbra{\psi'_r}{\psi'_r}.
\]
We see that
\[
F(U) = \sum_{b \in \{0,1\}^m} \tr(P_b Q_b)
\mbox{~~~~and~~~~}
F(U') = \sum_{b \in \{0,1\}^m} \tr(P_b Q'_b).
\]
Therefore
\begin{align*}
F(U) - F(U') &=  \sum_{b \in \{0,1\}^m} \tr(P_b (Q_b-Q'_b))\\
                 &\leq \sum_{b \in \{0,1\}^m} D_\text{tr}(Q_b,Q'_b)\\
                 &\leq D \sum_{r \in \{0,1\}^d} \frac{1}{D}\sqrt{1-|\braket{\psi_{r}}{\psi'_{r}}|^2}\\
                 &\leq \sqrt{D^2-\left(\sum_{r \in \{0,1\}^d} |\braket{\psi_{r}}{\psi'_{r}}|\right)^2}. 
\end{align*}
Here the first inequality follows from the variational characterization of trace distance ($D_\text{tr}(Q,Q')=\max_{P:\norm{P}\leq 1}\tr(P(Q-Q'))$); the second inequality follows from the convexity of trace distance, the fact that the $R_b$'s partition $\zone^d$, and a well-known expression for the trace distance of two pure states; and the third inequality follows from the concavity of the function $\sqrt{1-z^2}$. 

On the other hand, we can upper bound the Frobenius distance between $U$ and $U'$ by
\begin{align*}
    d(U, U') &= \sqrt{\sum_{r \in \{0,1\}^d} \norm{ \ket{\psi_{r}} - \ket{\psi'_{r}} }^2} = \sqrt{\sum_{r \in \{0,1\}^d} 2-2\Re(\braket{\psi_{r}}{\psi'_{r}})}\geq \sqrt{2D - 2\sum_{r \in \{0,1\}^d} | \braket{\psi_{r}}{\psi'_{r}}} |, 
\end{align*}
where the inequality uses the fact that $\Re(z) \leq |z|$ for any complex number~$z$. We find 
\begin{align*}
\frac{|F(U)-F(U')|}{d(U,U')} &\leq \sqrt{\frac{D^2- \left(\sum_{r \in \{0,1\}^d} |\braket{\psi_{r}}{\psi'_{r}}|\right)^2}{2D - 2\sum_{r \in \{0,1\}^d}| \braket{\psi_{r}}{\psi'_{r}}|}}\\
                             &= \sqrt{\frac{D+\sum_{r \in \{0,1\}^d} |\braket{\psi_{r}}{\psi'_{r}}|}{2}}\leq\sqrt{D},
\end{align*}
where the last inequality is because $|\braket{\psi_{r}}{\psi'_{r}}|\leq 1$ for each of the $D$ $r$'s, by Cauchy-Schwarz.
\end{proof}
For every pair of distinct inputs $x,y \in \{0,1\}^n$ and for every $\delta > 0$, it follows from the previous two results that the probability that the protocol's error probability on these inputs exceeds $2^{-m} + \delta$, is upper bounded by
\[
2\exp\left(\frac{-\delta^2 D^2}{4}\right)
\]
Setting $\delta = 2^{-m}$, $\eps= 2^{-m+1}$ and $d = \ceil{\frac{1}{2} \log_2 n + \log_2{\frac{1}{\eps}} + 4}$, by the union bound there is a positive probability that the resulting protocol has error probability at most $\eps$ for all input pairs. This implies the existence of the desired protocol, with $m=\ceil{\log 1/\delta}=\ceil{\log 1/\eps}+1$ bits of communication.

\section{Quantum one-way lower bound}\label{sec: 1way quantupure lower bound}
In this section we prove lower bounds on the one-way quantum communication complexity of any function whose communication matrix has a large number of distinct rows. 
As a consequence we obtain our lower bound for $\EQ_n$ of Theorem~\ref{thm: quant lower bound pure 1way}.

Let $F : \zone^n \times \zone^n \to \zone$ be a Boolean function.
We consider the model where communication is one-way, and Alice is only allowed to send a pure state to Bob.
Suppose there exists a protocol of cost $\log d$ that computes $F$ to error $\eps$.
Any such protocol looks like the following.
\begin{itemize}
    \item Alice, on input $x \in \zone^n$, sends a message $\ket{\phi_x}$ to Bob, where $\ket{\phi_x}$ is a unit vector in $\C^d$.
    \item Bob, on input $y$, measures with respect to projectors $P_y, I-P_y$.
\end{itemize}
The acceptance probability of the protocol is $\|P_y\ket{\phi_x}\|^2$. Thus, we have
\begin{align}\label{eqn: pxvxpyvy}
    \|P_y\ket{\phi_x}\|^2 \geq 1-\eps,\quad \|(I - P_y)\ket{\phi_x}\|^2 \leq \eps \quad \text{for all}~x, y \in F^{-1}(1),
\end{align}
and
\begin{align}\label{eqn: I-pxvxpyvy}
    \|P_y\ket{\phi_x}\|^2 \leq \eps,\quad \|(I - P_y)\ket{\phi_x}\|^2 \geq 1 - \eps \quad \text{for all}~x, y \in F^{-1}(0).
\end{align}

\begin{claim}\label{claim: 1way quantum inner product upper bound}
Let $F : \zone^n \times \zone^n\to\zone$ be a Boolean function with $N$ distinct rows in $M_F$. Let $X \subseteq \zone^n$ be an arbitrary subset of size $N$ that indexes distinct rows in $M_F$. For a one-way quantum communication protocol  as above that computes $F$ to error $\eps \leq 1/2$, we have
\[
2-2\sqrt{\eps(1 - \eps)}\leq \|\ket{\phi_{x_1}} - \ket{\phi_{x_2} }\|^2 \leq 2+4\sqrt{\eps}
\]
for all distinct $x_1,x_2 \in X$.
\end{claim}
\begin{proof}
Fix any two distinct $x_1,x_2 \in X$, and let $\ket{\phi_{x_1}}, \ket{\phi_{x_2}} \in \C^d$ be the messages sent by Alice on inputs $x_1, x_2$, respectively. Recall that $\|\ket{\phi_{x_1}}\| = \|\ket{\phi_{x_2}}\| = 1$.  Because of the assumption that the rows of $M_F$ indexed by $X$ are all distinct, there is a $y \in \zone^n$ such that $F(x_1, y) \neq F(x_2, y)$. Without loss of generality assume $F(x_1, y) = 1$ and $F(x_2, y) = 0$.
Write
\begin{align*}
    \ket{\phi_{x_1}} & = P_y \ket{\phi_{x_1}} + (I - P_y)\ket{\phi_{x_1}},\\
    \ket{\phi_{x_2}} & = P_y \ket{\phi_{x_2}} + (I - P_y)\ket{\phi_{x_2}}.
\end{align*}
Thus,
\begin{align*}
    \|\ket{\phi_{x_1}} - \ket{\phi_{x_2}}\|^2 & = \|P_y(\ket{\phi_{x_1}} - \ket{\phi_{x_2}})\|^2 + \|(I - P_y)(\ket{\phi_{x_1}} - \ket{\phi_{x_2}})\|^2 \tag*{since $P_y$ and $I - P_y$ are orthogonal projectors}\\
    & \geq (\|P_y\ket{\phi_{x_1}}\| - \|P_y\ket{\phi_{x_2}})\|)^2 + (\|(I - P_y)\ket{\phi_{x_1}}\| - \|(I - P_y)\ket{\phi_{x_2}}\|)^2 \tag*{by the triangle inequality}\\
    & \geq 2(\sqrt{1 - \eps} - \sqrt{\eps})^2 \tag*{by Equations~\eqref{eqn: pxvxpyvy} and~\eqref{eqn: I-pxvxpyvy}, and since $F(x_1, y) = 1$ and $F(x_2, y) = 0$}\\
    & = 2 - 2\sqrt{\eps(1 - \eps)}.
\end{align*}

For the upper bound, first define $p := \|P_y\ket{\phi_{x_1}}\|^2 \geq 1 - \eps$, and $q := \|(I - P_y)\ket{\phi_{x_2}}\|^2 \geq 1 - \eps$.
\begin{align*}
    \|\ket{\phi_{x_1}} - \ket{\phi_{x_2}}\|^2 & = \|P_y(\ket{\phi_{x_1}} - \ket{\phi_{x_2}})\|^2 + \|(I - P_y)(\ket{\phi_{x_1}} - \ket{\phi_{x_2}})\|^2\\
    & \leq (\|P_y\ket{\phi_{x_1}}\| + \|P_y\ket{\phi_{x_2}}\|)^2 + (\|(I - P_y)\ket{\phi_{x_1}}\| + \|(I - P_y)\ket{\phi_{x_2}}\|)^2 \tag*{by the triangle inequality}\\
    & = (\sqrt{p} + \sqrt{1-q})^2 + (\sqrt{1-p} + \sqrt{q})^2\\
    & = 2 + 2\sqrt{p(1-q)} + 2\sqrt{(1-p)q} \leq 2 + 4\sqrt{\eps}.
 \end{align*}
\end{proof}
We now state our main result of this section.
\begin{theorem}\label{thm: dist rows quant lb}
There exists an absolute constant $c$ such that the following holds.
Let $F : \zone^n \times \zone^n$ be a Boolean function with $N$ distinct rows in $M_F$. Then for all $\eps \in [1/N, 1/4]$,
\[
\Qpureoneway{\eps}(F) \geq \log \rbra{\frac{\log N}{\eps}} - \log\log \rbra{\frac{1}{\eps}} - c.
\]
\end{theorem}

\begin{proof}
Let $X \subseteq \zone^n$ be an arbitrary set of $N$ elements that index distinct rows in $M_F$.
Consider a protocol of cost $\log d$, as described in the beginning of this section, that computes $F$ to error $\eps$. Claim~\ref{claim: 1way quantum inner product upper bound} implies existence of vectors $\ket{\phi_x} \in \C^d$ for all $x \in X$, such that
\begin{equation}\label{eqn: phixs 2ish norm}
 2 - 2\sqrt{\eps(1 - \eps)} \leq \|\ket{\phi_{x_1}} - \ket{\phi_{x_2} }\|^2 \leq 2 + 4\sqrt{\eps}
\end{equation}
for all distinct $x_1,x_2 \in X$.
For each $x \in X$, define a real vector $\ket{\phi^R_x} \in \R^{2d}$ by
\[
\ket{\phi^R_x} = \sum_{j \in [d]}\ket{j}\left(R(\ket{\phi_x}_j) \ket{0} + C(\ket{\phi_x}_j) \ket{1}\right),
\]
where $R(\ket{\phi_x}_j)$ and $C(\ket{\phi_x}_j)$ denote the real and complex components of the $j$'th coordinate of $\ket{\phi_x}$, respectively. Note that each $\ket{\phi^R_x}$ is a unit vector, since the $\ket{\phi_x}$ are unit vectors.
For all distinct $x_1,x_2 \in X$, we have
\begin{align*}
\ket{\phi_{x_1}} - \ket{\phi_{x_2}} & = \sum_{j \in [d]} \ket{j}(R(\ket{\phi_{x_1}}_j - \ket{\phi_{x_2}}_j) + i \cdot C(\ket{\phi_{x_1}}_j - \ket{\phi_{x_2}}_j)),\\
\ket{\phi^R_{x_1}} - \ket{\phi^R_{x_2}} & = \sum_{j \in [d]} \ket{j}((R(\ket{\phi_{x_1}}_j - \ket{\phi_{x_2}}_j) \ket{0}) + (C(\ket{\phi_{x_1}}_j - \ket{\phi_{x_2}}_j) \ket{1})).
\end{align*}
Hence, Equation~\eqref{eqn: phixs 2ish norm} implies
\begin{align}\label{eqn: same inner products real complex}
    \|\ket{\phi^R_{x_1}} - \ket{\phi^R_{x_2}}\|^2 = \|\ket{\phi_{x_1}} - \ket{\phi_{x_2}}\|^2 \in [2 - 2\sqrt{\eps(1-\eps)},2+4\sqrt{\eps}]
\end{align}
for all distinct $x_1,x_2 \in X$.
Since $\norm{v- w}^2=\norm{v}^2+\norm{w}^2-2\inp{v}{w}$ for real vectors $v,w$, we obtain
\[
|\braket{\phi^R_{x_1}}{\phi^R_{x_2}}| \leq  2\sqrt{\eps}
\]
for all distinct $x_1,x_2 \in X$.
Now consider the $N \times N$ matrix $M$ whose rows and columns are indexed by strings in $X$, with entries defined by
\[
M_{x, y} = \braket{\phi^R_x}{\phi^R_y}.
\]
Since each $\phi^R_x \in \R^{2d}$, this matrix has rank at most $2d$. Since $\braket{\phi^R_x}{\phi^R_x} = 1$ for all $x \in \zone^n$ and $|\braket{\phi^R_x}{\phi^R_y}| \leq 2\sqrt{\eps}$ for all $x \neq y \in X$, this $M$ is a $2\sqrt{\eps}$-approximation to the $N \times N$ identity matrix~$I$.
Theorem~\ref{thm: alon} implies existence of an absolute constant $c_1 > 0$ such that 
\[
2d \geq \rank(M) \geq \rank_{2\sqrt{\eps}}(I) \geq \frac{c_1 \log N}{\eps\log(1/\sqrt{\eps})}. 
\]
Hence,
\[
\log d \geq \log\rbra{\frac{\log N}{\eps}} - \log\log\rbra{\frac{1}{\eps}} -\log(1/c_1),
\]
concluding the proof.
\end{proof}

Theorem~\ref{thm: quant lower bound pure 1way} immediately follows from Theorem~\ref{thm: dist rows quant lb} since all $2^n$ rows in $M_{\EQ_n}$ are distinct.

\section{Approximate-rank upper bounds for distributed SINK function}

In this section we show improved upper bounds on the approximate nonnegative-rank and approximate psd-rank of $M_{\SINK \circ \XOR}$, where $\SINK$ is defined as follows.

\begin{definition}\label{defn: sink}
Define the function $\SINK_n : \zone^n \to \zone$ on $n = \binom{m}{2}$ inputs as follows. The inputs are viewed as orientations of edges on a complete graph with $m$ vertices. The function outputs 1 if there is a sink in the graph, and 0 otherwise.
\end{definition}

Consider the function $\SINK_n \circ \XOR : \zone^{2n} \to \zone$. This function was recently used to refute the randomized and quantum versions of the log-rank conjecture~\cite{CMS20, SW19, ABT19}. Chattopadhyay, Mande and Sherif~\cite[Theorem 1.10]{CMS20} showed that the $1/3$-approximate rank of $M_{\SINK_n \circ \XOR}$ is $O(m^4)$ and the $1/3$-approximate nonnegative-rank of $M_{\SINK_n \circ \XOR}$ is $O(m^5)$. As a consequence of our improved upper bounds for the $\eps$-approximate nonnegative-rank of the Identity matrix (Corollary~\ref{cor: our apx nonneg rank identity upper bound}), we are able to use the same proof idea as theirs to obtain an $O(m^4)$ upper bound on the $1/3$-approximate nonnegative-rank of $M_{\SINK_n \circ \XOR}$, matching the approximate rank upper bound. We also obtain approximate psd-rank upper bounds for $\SINK_n \circ \XOR$.

\begin{claim}
Let $m$ be a positive integer, let $n = \binom{m}{2}$. Then,
\begin{align*}
\ranknonneg{1/3}(M_{\SINK_n \circ \XOR}) & = O(m^4)\\
\rankpsd{1/3}(M_{\SINK_n \circ \XOR}) & = O(m^{2.5}).
\end{align*}
\end{claim}
\begin{proof}
Note that $\SINK_n \circ \XOR$ can be expressed as a \emph{sum} of $m$ Equalities, each with $2(m-1)$ inputs, one corresponding to each vertex in the underlying graph for $\SINK$. Recall that the communication matrix of Equality is the Identity matrix. We require sub-additivity of nonnegative-rank and psd-rank, which are both easy to verify.

\begin{itemize}
    \item Corollary~\ref{cor: our apx nonneg rank identity upper bound} implies that each of these Equalities have $(1/3m)$-approximate nonnegative-rank $O(m^3)$. Summing up these $m$ matrices, we conclude that the $(1/3)$-approximate nonnegative-rank of $\SINK_n \circ \XOR$ equals $O(m^4)$.
    \item Corollary~\ref{cor: our apx psd rank identity upper bound} implies that each of these Equalities have $(1/3m)$-approximate psd-rank $O(m^{1.5})$. Summing up these $m$ matrices, we conclude that the $(1/3)$-approximate psd-rank of $\SINK_n \circ \XOR$ equals $O(m^{2.5})$.
\end{itemize}
\end{proof}

\section{Future work}

We mention some possible directions for future work:
\begin{itemize}
    \item Those of our lower bounds that use Alon's approximate-rank bound (Theorem~\ref{thm: alon}) lose an additive $\log\log(1/\eps)$. This term is necessary in some regimes, in particular when $\eps$ is very small ($\sim 2^{-n}$) and $n/\eps$ gets bigger than the trivial dimension upper bound $2^n$. However, in some regimes it may be avoidable.
    Also Alon's bound itself might be slightly improvable.
    \item We leave open the optimal quantum communication complexity of Equality with small error in the \emph{simultaneous message passing} (SMP) model, where Alice and Bob each send a message to a ``referee'' who has to decide the output. With public randomness $\log(1/\eps)\pm O(1)$ classical bits of communication are necessary and sufficient, but with private randomness it is not clear. 
     In the classical case,  $\Theta(\sqrt{n})$ bits of communication are necessary~\cite{newman&szegedy:1round} and sufficient~\cite{ambainis:3computer} for constant error.
    In the quantum case,  $\Theta(\log n)$ qubits are necessary and sufficient~\cite{bcww:fp} for constant error. One can get an $O(\log(n)\log(1/\eps))$ $\eps$-error upper bound by repeating the quantum fingerprinting protocol of Buhrman et al.~\cite{bcww:fp}
    $O(\log(1/\eps))$ times, but that is much worse than the $\log(\sqrt{n}/\eps)$ and $\log(n/\eps)$ upper bounds that we have in the one-way mixed-state and pure-state scenarios (Theorems~\ref{thm: quant upper bound mixed} and~\ref{thm: quant upper bound pure}). In neither the randomized nor the quantum SMP settings do we have tight bounds for small~$\eps$.
    \item We also leave open the optimal communication complexity of equality with small error in the entanglement-assisted setting. The classical public-coin protocol and the one we exhibited both require $\ceil{\log (1/\eps)} + O(1)$ bits of communication to compute $\EQ_n$ to within error $\epsilon$, and it seems probable that this is essentially optimal.
\end{itemize}

\paragraph{Acknowledgements.}
We thank Troy Lee, Ignacio Villanueva, and Zhaohui Wei for early discussions related to the result of Section~\ref{sec:upperboundmixed}. We thank Swagato Sanyal for discussions at an early stage of this work, from which the question of pinning down the exact communication complexity of Equality for small $\eps$ arose.

\bibliography{bibo}

\appendix

\section{Quantum communication complexity and psd-rank}\label{app: qepspsd}

In this section, we prove Theorem~\ref{thm: qeps apxpsd rank lower bound}, restated below.
\begin{theorem}[Restatement of Theorem~\ref{thm: qeps apxpsd rank lower bound}]\label{appthm: qeps apxpsd rank lower bound}
Let $F : \zone^n \times \zone^n \to \zone$ be a Boolean function and let $\eps > 0$. Then,
\[
\Q{\eps}(F) \geq \log \rankpsd{\eps}(M_F) + 1.
\]
\end{theorem}

\begin{proof}
Consider an $\ell$-qubit protocol for $F$, without public randomness.
Because private randomness can be generated using Hadamard gates, we will assume the protocol is unitary, with only a measurement of the output qubit at the end.
Let the starting state of the protocol be $\ket{x0^s}_A\ket{y0^s}_B\ket{0}_C$, where the first and second parts are Alice and Bob's register, respectively (containing their input and $s$ workspace qubits each), and the third part is the channel qubit. It is easy to prove by induction that after $\ell$ qubits of communication, the final state of a protocol has the following form (first observed by Kremer~\cite{kremer:thesis} and Yao~\cite{Yao93}):
\[
\sum_{i \in \zone^\ell}\ket{a_i(x)}\ket{b_i(y)}\ket{i_{\ell}},
\]
where $\ket{a_i(x)}, \ket{b_i(y)}$ are subnormalized quantum states.
Let $P$ denote the acceptance probability matrix, i.e., $P(x, y)$ is the probability that the protocol outputs 1 on input $(x, y)$. We assume without loss of generality that the output qubit is the last qubit put on the channel. We have
\[
P(x, y)  = \norm{\sum_{i \in \zone^\ell : i_{\ell = 1}}\ket{a_i(x)}\ket{b_i(y)}\ket{i_{\ell}}}^2 = \sum_{i, i' \in \zone^\ell : i_{\ell}=i'_{\ell}=1}  \braket{a_i(x)}{a_{i'}(x)}\cdot\braket{b_i(y)}{b_{i'}(y)}.
\]
For each $x \in \zone^n$ define a $2^{\ell-1} \times 2^{\ell-1}$ matrix $A_x$ with rows and columns indexed by strings  $i,i'\in\zone^{\ell-1}\times\{1\}$:
\[
A_x(i, i') = \braket{a_i(x)}{a_{i'}(x)}.
\]
Similarly, for each $y \in \zone^n$ define a $2^{\ell-1} \times 2^{\ell-1}$ matrix $B_y$  by 
\[
B_y(j, j') = \braket{b_j(y)}{b_{j'}(y)}.
\]
These $A_x$ and $B_y$ are Gram matrices and hence psd. Moreover it is easy to verify that $P(x, y) = \tr(A_xB_y)$.
Since the protocol makes error at most $\eps$ on each input, the matrix $P$ entrywise approximates $M_F$ up to~$\eps$. Hence $\rankpsd{\eps}(M_F)\leq 2^{\ell -1}$. Taking logarithms gives the theorem.
\end{proof}

\end{document}